\documentclass[11pt]{article}
\usepackage{amssymb,amsmath}
\usepackage{float}
\usepackage{color,fullpage}
\usepackage{amsthm}
\newtheorem{theorem}{Theorem}

\newtheorem{definition}{Definition}

\newcommand{\RR}{\mathbb{R}}
\newcommand{\st}{{\mathrm{~subject~to~}}}
\DeclareMathOperator*{\Min}{minimize}

\begin{document}

\title{A remark on weaken restricted isometry property in compressed sensing}

\author{Hui Zhang\thanks{
College of Science, National University of Defense Technology,
Changsha, Hunan, 410073, P.R.China. Corresponding author. Email: \texttt{h.zhang1984@163.com}}
}

\date{}

\maketitle

\begin{abstract}
The restricted isometry property (RIP) has become well-known in the compressed sensing community. Recently, a weaken version of
RIP was proposed for exact sparse recovery under weak moment assumptions. In this note, we prove that the weaken RIP is also sufficient for \textsl{stable and robust} sparse recovery by linking it with a recently introduced robust width property in compressed sensing. Moreover, we show that it can be widely apply to other compressed sensing instances as well.
\end{abstract}



\section{Introduction}\label{intro}
The concept of restricted isometry property (RIP), which was introduce by Cand$\grave{e}$s and Tao in \cite{candes2005decoding}, has become well-known in the compressed sensing community. Recently, a weaken version of RIP appeared in \cite{lecue2014sparse} as an important tool for \textsl{exact} sparse recovery under weak moment assumptions. In this note, we prove that the weaken RIP is also sufficient for \textsl{stable and robust} sparse recovery. A key observation is that the weaken RIP implies the robust width property, which was proposed in \cite{jameson2014robust} and allows uniformly stable and robust recovery for many instances in compressed sensing \cite{jameson2014robust,zhang2015robust}. We establish our discovery in a very general form by introducing the concepts of atom norm and simple subset, and hence other instances in compressed sensing such as low-rank recovery are commonly covered.  The results in this study indicate that the weaken RIP plays a vital role similar to the original RIP.

\section{Robust width and weaken-RIP}
\subsection{Robust width property}
The robust width property was formally defined in \cite{jameson2014robust}. Before stating its definition, we first introduce some notations. Let $\mathcal{H}$ be a finite-dimensional Hilbert space with inner product $\langle \cdot, \cdot\rangle$ and $\|\cdot\|$ be the norm induced by the inner product over  $\mathcal{H}$. We denote the Euclidean norm by $\|\cdot\|_2$. Let $\Phi :\mathcal{H}\rightarrow \mathbb{F}^m$ denote some known linear operator, where $\mathbb{F}$ is either $\mathbb{R}$ or $\mathbb{C}$.
\begin{definition}[robust width property, \cite{jameson2014robust}]
We say a linear operator $\Phi :\mathcal{H}\rightarrow \mathbb{F}^m$ satisfies the $(\rho,\alpha)$-robust width property over $B_\sharp$ if
$$ \|x\|\leq \rho \|x\|_\sharp$$
for every $x\in \mathcal{H}$ such that $\|\Phi x\|_2<\alpha \|x\|$; or equivalently if
$$\|\Phi x\|_2\geq \alpha \|x\| $$
for every $x\in \mathcal{H}$ such that $\|x\|>\rho \|x\|_\sharp$. Here, $\|\cdot\|_\sharp$ is certain norm used to promote certain structured solutions to underdetermined systems of linear equations.
\end{definition}

\subsection{Weaken-RIP property}
The RIP was originally introduced in \cite{candes2005decoding}. In the latter, many researchers have contributed to the topic of RIP; for more information on its development, one could refer to \cite{foucart2014math}. Here, we write down the definition of RIP formulated in \cite{cohen2009compressed}.
\begin{definition}[RIP, \cite{candes2005decoding,cohen2009compressed}]
Let $\Sigma_k:=\{x\in \RR^n, \|x\|_0\leq k\}$ be the k-sparse vector set, where $\|x\|_0$ stands for the number of nonzero
components of vector $x$. A matrix $\Gamma\in\RR^{m\times n} $ is said to have $(k,\delta)$-RIP property for sparsity $k\in [n]:=\{1, 2,\cdots, n\}$ and distortion $0<\delta<1$ if
for every $k$-sparse vector in $\Sigma_k$, it holds
$$(1-\delta)\|x\|_2\leq \|\Gamma x\|_2\leq (1+\delta) \|x\|_2.$$
 \end{definition}

The following is the definition of the weaken RIP appeared in \cite{lecue2014sparse}.
\begin{definition}[weaken-RIP,  \cite{lecue2014sparse}]
Let $(e_1, \cdots, e_n)$ be the canonical
basis of $\RR^n$. A matrix $\Gamma\in\RR^{m\times n} $ is said to have  $(k,\alpha, \beta)$-weaken-RIP property for sparsity $k\in [n]$ and distortions $\alpha, \beta>0$ if
\begin{enumerate}
  \item[a)] for every $k$-sparse vector in $\Sigma_k$, $\|\Gamma x\|_2 \geq \alpha \|x\|_2$ and
  \item[b)] for every $i\in [n]$,  $\|\Gamma e_i\|_2 \leq \beta$.
\end{enumerate}
 \end{definition}

Compared with the original RIP, the weaken-RIP are strictly weaker, as it suffices to verify the right-hand side of the RIP just for 1-sparse vectors and not for all s-sparse vectors \cite{lecue2014sparse}. Because of such relaxed restriction of the right-hand side of the RIP, the weaken-RIP contributes as a main reason for weakening moment assumptions in exact sparse recovery. In this note, we will show that the weaken-RIP also ensures stable and robust sparse recovery. Moreover, the idea behind of the weaken-RIP can widely apply to other instances in compressed sensing. In what follows, we try to state a generalized version of the weaken-RIP. To this end, let $\mathcal{B}$  be a collection of atoms such that the following holds
$$\mathcal{H}=\{x\in \mathcal{H}: x=\sum_{b \in\mathcal{B}}  \sigma_b  b , \sigma_b \geq 0, \forall b \in\mathcal{B}\}.$$
We assume that $\mathcal{B}$ is centrally symmetric about the origin and consider the atom norm $\|\cdot\|_{\mathcal{B}}$ induced by $\mathcal{B}$. Recall that $\|\cdot\|_{\mathcal{B}}$ has the following expression \cite{chandrasekaran2012convex}:
$$\|x\|_{\mathcal{B}}= \inf \{ \sum_{b \in\mathcal{B}}  \sigma_b: x=\sum_{b \in\mathcal{B}}  \sigma_b  b , \sigma_b \geq 0, \forall b \in\mathcal{B}\}.$$
Let $k<n$ and define $k$-simple subset relative to $\mathcal{B}$ as follows:
$$\mathcal{A}:=\{x\in \mathcal{H}: x=\sum_{i=1}^k \sigma_i b_i, \sigma_i\geq 0, b_i\in\mathcal{B}\}.$$
Before going forward, we give two well-known instances in compressed sensing. The first one is the sparse vector recovery, where we can take $\mathcal{B}=\{\pm e_1, \pm e_2,\cdots, \pm e_n\}$, and then the corresponding  $k$-simple subset $\mathcal{A}$ relative to $\mathcal{B}$ is just the k-sparse vector set $\Sigma_k$, and the induced norm by $\mathcal{B}$ is the $\ell_1$-norm. The second one is the low-rank matrix recovery, where we can take $\mathcal{B}=\{uv^T\in\RR^{n\times n}:\|u\|_2=\|v\|_2=1\}$, and then the corresponding  $k$-simple subset $\mathcal{A}$ relative to $\mathcal{B}$ is k-rank matrix set,  and the induced norm by $\mathcal{B}$ is the Schatten $\ell_1$-norm. For more cases, please refer to \cite{chandrasekaran2012convex}.
\begin{definition}[Generalized-weaken-RIP]
Let $\mathcal{A}$ be a $k$-simple subset relative to $\mathcal{B}$. A linear operator $\Phi :\mathcal{H}\rightarrow \mathbb{F}^m$ is said to have  $(k,\alpha, \beta)$-generalized-weaken-RIP property for sparsity $k\in [n]$ and distortions $\alpha, \beta>0$ if
\begin{enumerate}
  \item[a)] for every element in $\mathcal{A}$, $\|\Phi x\|_2 \geq \alpha \|x\| $ and
  \item[b)] for every element in $\mathcal{B}$, $\|\Phi x\|_2 \leq \beta \|x\|$.
\end{enumerate}
 \end{definition}
 Note that the above definition on one hand is a generalization of the weaken-RIP, and on the other hand can also be viewed as a weaken version of the generalized RIP in \cite{bourrier2015fund}.
\begin{definition}[Generalized-RIP, \cite{bourrier2015fund}]
Let $\mathcal{E}$ be some subset of $\mathcal{H}$. A linear operator $\Phi :\mathcal{H}\rightarrow \mathbb{F}^m$ is said to have  $(k,\alpha, \beta)$-generalized-RIP property for sparsity $k\in [n]$ and distortions $\alpha, \beta>0$ if
\begin{enumerate}
  \item[a)] for every element in $\mathcal{E}$, $\|\Phi x\|_F \geq \alpha \|x\|_G $ and
  \item[b)] for every element in $\mathcal{E}$, $\|\Phi x\|_F \leq \beta \|x\|_G$
\end{enumerate}
where $\|\cdot\|_F$ and $\|\cdot\|_G$ are certain norms.
 \end{definition}

\section{Main results}
Now, we state our main result:
\begin{theorem}[weaken-RIP implies robust width property]\label{mainthm}
Let $\mathcal{A}$ be a $k$-simple subset relative to  $\mathcal{B}$, assume that each element $b\in \mathcal{B}$ satisfies $\|b\|=1$, and take $\|\cdot\|_{\sharp}=\|\cdot\|_{\mathcal{B}}$. If linear operator $\Phi :\mathcal{H}\rightarrow \mathbb{F}^m$ has the $(k,\alpha, \beta)$-generalized-weaken-RIP property for sparsity $1 < k\in [n]$ and distortions $\alpha, \beta>0$, then it must have the $(\rho, \widetilde{\alpha})$-robust width property with
$$\widetilde{\alpha}=\sqrt{ \alpha^2 -\frac{\beta^2-\alpha^2}{\rho^2(k-1)}}.$$
\end{theorem}

\begin{proof}
We divide the proof into two steps.

\textbf{Step 1:} Let $x$ be any nonzero element in $\mathcal{H}$; then it has representations via the elements in $\mathcal{B}$. Take one of its representations, say $x=\sum_{b \in\mathcal{B}}  \sigma_b  b , \sigma \geq 0, b \in\mathcal{B}$. Let $\gamma= \sum_{b \in\mathcal{B}}  \sigma_b$; it must be positive since $x$ is a nonzero element.
Now, we prove the following inequality via the Maurey empirical method, which was employed to prove similar results in \cite{oliveira2013lower,lecue2014sparse}:
\begin{equation}\label{ineq1}
\|\Phi x\|_2^2\geq \alpha^2 \|x\|^2-\frac{\gamma}{k-1}\left(\sum_{b \in\mathcal{B}} \sigma_b\|\Phi b\|_2^2-\alpha^2 \sum_{b \in\mathcal{B}} \sigma_b\|b\|^2\right).
\end{equation}

Let $Y$ be a random element in $\mathcal{H}$ defined by
$$\mathbb{P}(Y=\gamma b)=\frac{\sigma_b}{\gamma},$$
where $b \in\mathcal{B}$. Then, $$\mathbb{E} Y=  \sum_{b \in\mathcal{B}} \gamma b \frac{\sigma_b}{\gamma}=\sum_{b \in\mathcal{B}}  \sigma_b  b =x.$$
Let $Y_1, Y_2, \cdots, Y_k$ be independent copies of $Y$ and set $z=\frac{1}{k}\sum_{i=1}^kY_i$. Then, $z\in \mathcal{A}$ for every realization of $Y_1, Y_2, \cdots, Y_k$. By the definition of the $(k,\alpha, \beta)$-generalized-weaken-RIP property, we have $\|\Phi z\|_2^2\geq \alpha^2\|z\|^2$ and hence
$$\mathbb{E} \|\Phi z\|_2^2\geq \alpha^2\mathbb{E}\|z\|^2.$$
It is straightforward to verify the following relationships:
\begin{enumerate}
  \item [i)] $\mathbb{E} \langle Y, Y\rangle =\gamma  \sum_{b \in\mathcal{B}} \sigma_b\|b\|^2$,
  \item [ii)] for every $1\leq i\leq k$, $$\mathbb{E} \langle \Phi Y_i,\Phi  Y_i\rangle =\gamma  \sum_{b \in\mathcal{B}} \sigma_b\|\Phi b\|^2,$$
  \item [iii)] and for each pair $(i,j)$ satisfying $1\leq i\neq j\leq k$, $$\mathbb{E} \langle \Phi Y_i,\Phi  Y_j\rangle = \|\Phi x\|_2^2.$$
\end{enumerate}
Therefore, we have
$$\mathbb{E}  \|\Phi z\|_2^2=\frac{1}{k^2}\sum_{i,j=1}^k\mathbb{E}\langle \Phi Y_i, \Phi Y_j\rangle
=\frac{k-1}{k}\|\Phi x\|_2^2+\frac{\gamma}{k}\sum_{b \in\mathcal{B}} \sigma_b\|\Phi b\|^2.$$
Similarly, it holds that
$$\mathbb{E}  \| z\|^2=\frac{1}{k^2}\sum_{i,j=1}^k\mathbb{E}\langle  Y_i, Y_j\rangle=\frac{k-1}{k}\|x\|_2^2+\frac{\gamma}{k}\sum_{b \in\mathcal{B}} \sigma_b\|b\|^2.$$
Substituting these two relationships to $\mathbb{E} \|\Phi z\|_2^2\geq \alpha^2\mathbb{E}\|z\|^2$ and rearranging terms, we obtain the desired inequality.

\textbf{Step 2:} Apply the upper bound estimation in weaken-RIP and the robust width property to finish the proof. Notice that $\|b\|=1$ and $\|\Phi b\|_2^2\leq \beta^2\|b\|^2=\beta^2$, the inequality \eqref{ineq1} can be simplified into
\begin{equation}\label{ineq2}
\|\Phi x\|_2^2\geq \alpha^2 \|x\|^2-\frac{\gamma^2(\beta^2-\alpha^2)}{k-1}
\end{equation}
Since the expression of $x$ is taken arbitrarily, we derive that
\begin{subequations}
\begin{align*}
\|\Phi x\|_2^2& \geq \sup_{x=\sum_{b \in\mathcal{B}}  \sigma_b  b, \sigma_b\geq 0}\left(\alpha^2 \|x\|^2-\frac{\gamma^2(\beta^2-\alpha^2)}{k-1}\right)  \\
  &= \alpha^2 \|x\|^2-\frac{\beta^2-\alpha^2}{k-1} \sup_{x=\sum_{b \in\mathcal{B}}  \sigma_b  b, \sigma_b\geq 0} (\sum_{b \in\mathcal{B}}\sigma_b)^2  \\
  & = \alpha^2 \|x\|^2-\frac{\|x\|^2_{\mathcal{B}}(\beta^2-\alpha^2)}{k-1},
\end{align*}
\end{subequations}
where the last relationship follows from the expression of atom norm $\|\cdot\|_{\mathcal{B}}$. Therefore, for every $x\in\mathcal{H}$ satisfying $\|x\|>\rho \|x\|_{\mathcal{B}}$ we get
$$\|\Phi x\|_2^2  \geq \left[\alpha^2 -\frac{\beta^2-\alpha^2}{\rho^2(k-1)}\right]\|x\|^2,$$
which implies the $(\rho, \widetilde{\alpha})$-robust width property.  This completes the proof.
\end{proof}

It has been shown that the $(\rho,\alpha)$-robust width property is sufficient (and necessary up to constants) for uniformly stable and robust sparse recovery by convex minimization. Under the same setting of Theorem \ref{mainthm}, the $(k,\alpha, \beta)$-generalized-weaken-RIP property is a stronger property and hence ensures uniformly \textsl{stable and robust} sparse recovery as well. In the following, we first introduce the concept of compressed sensing space, and then state a group of uniformly stable and robust sparse recovery results without proof details; they can be obtained directly by combining Theorem \ref{mainthm} with the stable and robust results in \cite{jameson2014robust,zhang2015robust}.

\begin{definition}(\cite{jameson2014robust})
A compressed sensing space $(\mathcal{H},\mathcal{A}, \|\cdot\|_{\sharp} )$  with bound $L$ consists of a finite-dimensional Hilbert space $\mathcal{H}$, a subset $\mathcal{A}\subseteq \mathcal{H}$, and a norm $\|\cdot\|_{\sharp}$ on $\mathcal{H}$ with following properties:

 (i) $0\in \mathcal{A}$.

 (ii) For every $a\in \mathcal{A}$ and $v\in \mathcal{H}$, there exists a decomposition $v= z_1+z_2$ such that
 $$\|a+z_1\|_{\sharp} =\|a\|_{\sharp}+\|z_1\|_{\sharp}, ~~~~\|z_2\|_{\sharp}\leq L\|v\|_2.$$
\end{definition}

 \begin{theorem}[weaken-RIP implies stable and robust recovery]
Let $\mathcal{A}$ be a $k$-simple subset relative to $\mathcal{B}$ and assume that each element $b\in \mathcal{B}$ satisfies $\|b\|=1$. Let $\|\cdot\|_{\mathcal{B}}$ be the induced norm by $\mathcal{B}$ and $\|\cdot\|_{\diamond}$ be its dual norm.  Suppose that there exists a bound $L$ such that the triple $(\mathcal{H},\mathcal{A}, \|\cdot\|_{\mathcal{B}} )$ is a compressed sensing space.  If linear operator $\Phi :\mathcal{H}\rightarrow \mathbb{F}^m$ has the $(k,\alpha, \beta)$-generalized-weaken-RIP property for sparsity $1 < k\in [n]$ and distortions $\alpha, \beta>0$, then
\begin{enumerate}
  \item [a)] for every $x^\natural \in\mathcal{H}, \theta\in (0,1), \lambda>0, \sigma>0$ and $w\in\mathbb{F}^M$ satisfying $\|\Phi^Tw\|_\diamond\leq \theta\lambda\sigma$, any solution $x^*$ to the Lasso model
$$\Min \frac{1}{2}\|\Phi x-(\Phi x^\natural +w)\|_2^2+\lambda\sigma \|x\|_{\mathcal{B}}$$
satisfies $\|x^*-x^\natural \| \leq C_0 \|x^\natural -a \|_{\mathcal{B}} +C_1 \cdot \sigma$ for every $a\in \mathcal{A}$. Here, $\lambda$ is some turning parameter and $\sigma$ is a measurement of the noise level, and
$$C_0=\left( \frac{1-\theta}{2\rho}-L\right)^{-1},~~~~ C_1=\frac{(1+\theta)\lambda}{\widetilde{\alpha}^2\rho}$$
where $\widetilde{\alpha}$ is given by Theorem \ref{mainthm} and the parameter $\rho$ satisfies $$\sqrt{\frac{\beta^2-\alpha^2}{\alpha^2(k-1)}}<\rho<\frac{1-\theta}{2L}.$$

\item [b)] for every $x^\natural \in\mathcal{H}$ and $e\in\mathbb{F}^M$  with  $\|e\|_2\leq \epsilon$, any solution $x^*$ to the basis pursuit model
$$\Min \|x\|_{\mathcal{B}}, ~\st ~~\|\Phi x-(\Phi x^\natural +w)\|_2\leq \epsilon $$
satisfies $\|x^*-x^\natural \| \leq C_2 \|x^\natural -a \|_{\mathcal{B}} +C_3 \cdot \sigma$ for every $a\in \mathcal{A}$. Here,
$$C_2=2\rho,~~~~ C_3=\frac{2}{\widetilde{\alpha}}$$
where the parameter $\rho$ satisfies $$\sqrt{\frac{\beta^2-\alpha^2}{\alpha^2(k-1)}}<\rho<\frac{1}{4L}.$$

  \item [c)]for every $x^\natural \in\mathcal{H}, \lambda>0, \sigma>0$ and $w\in\mathbb{F}^M$ satisfying $\|\Phi^Tw\|_\diamond\leq \lambda\sigma$, any solution $x^*$ to the Dantzig selector model
$$\Min \|x\|_{\mathcal{B}}, ~\textrm{subject to}~~  \|\Phi^T(\Phi x-(\Phi x^\natural +w))\|_{\diamond}\leq \lambda\sigma$$
satisfies $\|x^*-x^\natural \| \leq C_4 \|x^\natural -a \|_\sharp +C_5 \cdot \sigma$ for every $a\in \mathcal{A}$.
Here,
$$C_4=\left( \frac{1}{2\rho}-L\right)^{-1},~~~~ C_5=\frac{2\lambda}{\tilde{\alpha}^2\rho}$$
where the parameter $\rho$ satisfies $$\sqrt{\frac{\beta^2-\alpha^2}{\alpha^2(k-1)}}<\rho<\frac{1}{2L}.$$
\end{enumerate}
\end{theorem}

\section*{Acknowledgement}
 The work is supported by  the National Science Foundation of China (No.61271014 and No.61072118).

%
%


\end{document}